\documentclass[10pt, conference, letterpaper]{IEEEtran}

\IEEEoverridecommandlockouts

\usepackage{graphicx}
\usepackage{bbm}
\usepackage{amsmath}
\usepackage{amsthm}
\usepackage{indentfirst}
\usepackage[ruled,vlined]{algorithm2e}

\usepackage{hyperref}
\hypersetup{
	colorlinks=true, 
    linkcolor=blue,  
    citecolor=blue,  
    filecolor=blue,  
    urlcolor=blue    
}

\newtheorem{theorem}{Theorem}
\newtheorem{corollary}[theorem]{Corollary}
\newtheorem{proposition}[theorem]{Proposition}
\newtheorem{lemma}[theorem]{Lemma}
\newtheorem{remark}[theorem]{Remark}
\newtheorem{definition}[theorem]{Definition}

\DeclareMathOperator*{\argmax}{arg\,max}
\DeclareMathOperator*{\argmin}{arg\,min}

\begin{document}
\title{Aging Bandits: Regret Analysis and Order-Optimal Learning Algorithm for Wireless Networks with Stochastic Arrivals}

\author{Eray Unsal Atay, Igor Kadota, and Eytan Modiano}

\maketitle

\begin{abstract}
	
We consider a single-hop wireless network with sources transmitting time-sensitive information to the destination over multiple unreliable channels. Packets from each source are generated according to a stochastic process with known statistics and the state of each wireless channel (ON/OFF) varies according to a stochastic process with unknown statistics. The reliability of the wireless channels is to be learned through observation. At every time slot, the learning algorithm selects a single pair (source, channel) and the selected source attempts to transmit its packet via the selected channel. The probability of a successful transmission to the destination depends on the reliability of the selected channel. The goal of the learning algorithm is to minimize the Age-of-Information (AoI) in the network over $T$ time slots. To analyze the performance of the learning algorithm, we introduce the notion of AoI regret, which is the difference between the expected cumulative AoI of the learning algorithm under consideration and the expected cumulative AoI of a genie algorithm that knows the reliability of the channels a priori. The AoI regret captures the penalty incurred by having to learn the statistics of the channels over the $T$ time slots. The results are two-fold: first, we consider learning algorithms that employ well-known solutions to the stochastic multi-armed bandit problem (such as $\epsilon$-Greedy, Upper Confidence Bound, and Thompson Sampling) and show that their AoI regret scales as $\Theta(\log T)$; second, we develop a novel learning algorithm and show that it has $O(1)$ regret. To the best of our knowledge, this is the first learning algorithm with bounded AoI regret.

\end{abstract}

\section{Introduction}\label{sec.Introduction}

Age-of-Information (AoI) is a performance metric that captures the freshness of the information from the perspective of the destination. AoI measures the time that elapsed since the generation of the packet that was most recently delivered to the destination. This performance metric has been receiving attention in the literature \cite{Book_AoI_17,Book19,yates2020age} for its application in communication systems that carry time-sensitive data. In this paper, we consider a network with $M$ sources transmitting time-sensitive information to the destination over $N$ unreliable wireless channels, as illustrated in Fig.~\ref{fig.Network}. 
Packets from each source are generated according to an i.i.d.\ stochastic process with known statistics and the state of each wireless channel (ON/OFF) varies according to an i.i.d.\ stochastic process with \emph{unknown statistics}. At every time slot, the learning algorithm schedules a single pair (source, channel) and the selected source attempts to transmit its packet via the selected wireless channel. When a packet with fresh information is successfully transmitted to the destination, the AoI associated with the selected source is reduced. The goal of the scheduler is to keep the information associated with every source in the network as fresh as possible, i.e.\ to minimize the AoI in the network. To decide which pair to select in a time slot, the scheduler takes into account: i) the packet generation processes at the $M$ sources; ii) the current values of AoI at the destination; and iii) the estimated reliability of the $N$ wireless channels. 

In this sequential decision problem, the outcomes of previous transmission attempts are used to estimate the reliability of the wireless channels. This statistical learning problem is closely related to the stochastic multi-armed bandit (MAB) problem in which the wireless channels are the bandits that give i.i.d.\ rewards and the scheduler is the player that attempts to learn the statistics of the bandits in order to maximize the reward accumulated over time. The main challenge in the stochastic MAB problem is to strike a balance between exploiting the bandit that gave the highest rewards in the past and exploring other bandits that may give high rewards in the future. To evaluate the performance of different learning algorithms, we define regret. Regret is the difference between the expected cumulative reward of a \emph{genie algorithm} (that knows the statistics of the bandits a priori) and the expected cumulative reward of the \emph{learning algorithm} under consideration. The regret captures the penalty incurred by having to learn the statistics of the bandits over time. Some well-known order-optimal learning algorithms in terms of regret are: $\epsilon$-Greedy, Upper Confidence Bound (UCB), and Thompson Sampling (TS). The regret of these policies was shown to increase no more than logarithmically in time \cite{FiniteMAB,KL-UCB,TS_MAB}, $O(\log T)$, and this bound was shown to be tight \cite{Assympt}. 

We refer to our problem as the \emph{Aging Bandit problem}. An important distinction between the stochastic MAB problem and the Aging Bandit problem is the reward structure. In the stochastic MAB problem, the player selects a bandit in each time slot and receives a reward that is i.i.d.\ over time and depends only on the probability distribution associated with the selected bandit. In the Aging Bandit problem, the scheduler selects a pair (source, channel) and the reward is the AoI reduction that results from a packet transmission to the destination. This reward depends on the state of the selected channel (which is i.i.d.\ over time), since a failed transmission gives zero reward, and it also depends on the history of previous packet deliveries and packet generations. In particular, if the selected source has recently delivered a fresh information update to the destination, then the reduction in AoI may be small. In contrast, if the selected source has not updated the destination for a long period, then the AoI reduction may be large. The reward structure of Aging Bandits is closely related to the AoI evolution (formally defined in Sec.~\ref{sec.Model}) which is history-dependent. This intricate reward structure has significant impact on the analysis of regret and on the development of learning algorithms when compared to the analysis of the traditional stochastic MAB. 

The literature on MAB problems is vast, dating more than eight decades \cite{Thompson33}. For surveys on different types of MAB problems, we refer the readers to \cite{surveyMAB1,RMAB_book,surveyMAB2,surveyMAB3}. Most relevant to this work are \cite{RegretQueueing,RegretThomas,ThomasPHD,igorTON18,igorINFOCOM,YuPin18,YuPinTMC,WhittleCSMA,WhittleVishrant,RegretAoI,CorrelatedRegretAoI,DecentralizedRegretAoI}. The authors in \cite{RegretQueueing,RegretThomas,ThomasPHD} considered the problem of minimizing the expected queue-length in a system with a single queue and multiple servers with unknown service rates. In \cite{RegretQueueing}, the authors introduced the concept of queue-length regret, developed a learning algorithm inspired by Thompson Sampling, and analyzed its regret. In \cite{RegretThomas,ThomasPHD}, the authors used information particular to the queue evolution to develop a learning algorithm with $O(1)$ queue-length regret. 

The authors in \cite{igorTON18,igorINFOCOM,YuPin18,YuPinTMC,WhittleCSMA,WhittleVishrant,RegretAoI,CorrelatedRegretAoI,DecentralizedRegretAoI} considered the problem of minimizing the average AoI in a single-hop wireless network with unreliable channels. In \cite{igorTON18,igorINFOCOM,YuPin18,YuPinTMC,WhittleCSMA,WhittleVishrant}, the authors posed the AoI minimization problem in a network with multiple sources and \emph{known channel statistics} as a restless MAB problem, developed the associated Whittle’s Index scheduling policy, and evaluated its performance in terms of the average AoI. In \cite{RegretAoI}, the authors considered the AoI minimization problem in a network with a single source-destination pair and unknown channel statistics, introduced the concept of AoI regret, and showed that the AoI regret of UCB and TS scale as $O(\log T)$. In \cite{CorrelatedRegretAoI}, the authors obtained similar results as in \cite{RegretAoI} for the more challenging case of correlated wireless channels. In \cite{DecentralizedRegretAoI}, the authors considered the AoI minimization problem in a network with multiple sources that generate and transmit fresh packets at every time slot through (possibly) different channels with unknown statistics. The authors in \cite{DecentralizedRegretAoI} showed that the AoI regret of a UCB-based distributed learning algorithm scales as $O(\log^2 T)$. An important modelling assumption common to \cite{RegretAoI,CorrelatedRegretAoI,DecentralizedRegretAoI} is that sources generate and transmit fresh packets at every time slot. The more realistic assumptions of random packet generation and scheduled transmissions have significant impact on the AoI evolution, on the analysis of AoI regret, and on the development of learning algorithms. For example, in Sec.~\ref{sec.Policy}, we leverage the random packet generation to develop a learning algorithm with $O(1)$ AoI regret.

In this paper, we study learning algorithms that attempt to minimize AoI in a network with multiple sources generating packets according to stochastic processes and transmitting these packets to the destination over wireless channels with initially unknown statistics. At every time slot, the learning algorithm schedules a single pair (source, channel) and the selected source attempts to transmit a packet through the selected channel. Note that the source policy, which selects a source at each time slot, and the channel policy, which selects the channel to be used in each time slot, can be naturally decoupled, as the optimal channel is independent of the source selected. 
In this paper, we focus on the exploration-exploitation dilemma faced by the channel policy. In particular, we consider learning algorithms employing the \emph{optimal source policy} and \emph{different channel policies}. Our main contributions include: 
\begin{itemize}
	\item we analyze the performance of channel policies based on traditional MAB algorithms including $\epsilon$-Greedy, UCB, and TS, and show that their AoI regret scales as $\Theta(\log T)$. These results generalize the analysis in \cite{RegretAoI} to networks with multiple sources generating packets randomly. The analysis of the AoI regret is more challenging in this network setting since the AoI evolution depends on both the source policy and the stochastic packet generation process. These challenges are discussed in Sec.~\ref{sec.Regret};
	\item we develop a novel learning algorithm and establish that it has $O(1)$ AoI regret. The key insight is that when packets are generated randomly, the learning algorithm can utilize times when the network has no packets to transmit, in order to learn the statistics of the channel. To the best of our knowledge, this is the first learning algorithm with bounded AoI regret.
\end{itemize}

The remainder of this paper is outlined as follows. In Sec.~\ref{sec.Model}, the network model and performance metrics are formally presented. In Sec.~\ref{sec.Regret}, we analyze the AoI regret of traditional learning algorithms. In Sec.~\ref{sec.Policy}, we develop an order-optimal learning algorithm and analyze its AoI regret. In Sec.~\ref{sec.Simulations}, we compare the AoI regret of different learning algorithms using simulations. The paper is concluded in Sec.~\ref{sec.Conclusion}. Some of the technical proofs have been omitted due to the space constraint, and will be made available in a technical report.

\section{System Model}\label{sec.Model}

Consider a single-hop wireless network with $M$ sources, $N$ channels and a single destination, as illustrated in Fig.~\ref{fig.Network}. 
Each source generates packets containing time-sensitive information and these packets are to be transmitted to the destination through one of the wireless channels. Let the time be slotted, with slot index $t\in\{1,2,\cdots,T\}$, where $T$ is the time horizon of this discrete-time system. The slot duration allows for a single packet transmission. We normalize the slot duration to unity.

\begin{figure}[t]
\begin{center}
\includegraphics[width=0.9\columnwidth]{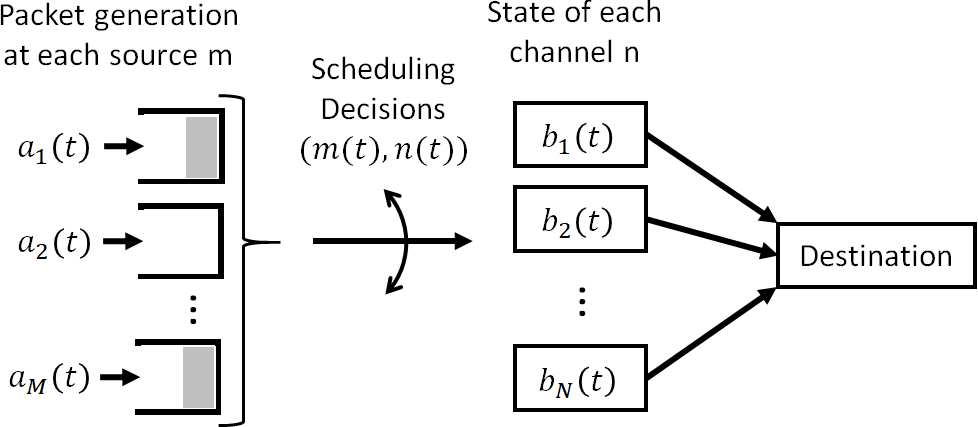}
\end{center}
\vspace{-0.4cm}
\caption{Illustration of the wireless network with $M$ sources, $N$ channels, and a destination.}\label{fig.Network}
\vspace{-0.4cm}
\end{figure}

At the beginning of every slot $t$, each source generates a packet with probability $\lambda\in(0,1)$. Let $a_m(t)\in\{0,1\}$ be the indicator function that is equal to $1$ when source $m\in\{1,2,\cdots,M\}$ generates a packet in slot $t$, and $a_m(t)=0$ otherwise. This Bernoulli process with parameter $\lambda$ is i.i.d.\ over time and independent across different sources, with $\mathrm{P}\left(a_m(t)=1\right)=\lambda,\forall m,t$. A packet that is generated in slot $t$ can be transmitted during the same slot $t$. We denote the vector of packet generations in slot $t$ by $\vec{a}(t) = \left[ a_1(t)\ \cdots\ a_M(t) \right]^\mathsf{T}$.

Each source has a transmission queue to store its packets. Sources keep \emph{only} the most recently generated packet, i.e.\ the freshest packet, in their queue. When source $m$ generates a new packet at the beginning of slot $t$, older packets (if any) are discarded from its queue. Notice that delivering the most recently generated packet provides the freshest information to the destination. This queueing discipline is known to optimize the AoI in a variety of contexts \cite{AoI_management,AoI_LIFO,AoI_LGFS19_2}. After a packet delivery from source $m$, the queue remains empty until the next packet generation from the same source. However, while the queue is empty, a \emph{dummy packet} can be transmitted for the purpose of probing the channels.

The networked system is empty during slot $t$ if there are no data packets available for transmission, i.e.\ if the $M$ queues are empty. Let $E(t)\in\{0,1\}$ be the indicator function that is equal to $1$ if the system is empty during slot $t$, and $E(t)=0$ otherwise. Notice that if there is a packet generation at the beginning of slot $t$, then the system is nonempty during slot $t$ and $E(t)=0$. Recall that when the system is empty, sources can still transmit dummy packets.


In a slot, the learning algorithm selects a single pair $(m,n)$, where $m\in\{1,2,\cdots,M\}$ is the index of the source and $n\in\{1,2,\cdots,N\}$ is the index of the wireless channel. Then, during this slot, source $m$ transmits a packet to the destination through channel $n$. If channel $n$ is ON, then the packet is successfully transmitted to the destination, and if channel $n$ is OFF, then the transmission fails. The learning algorithm does not know the channel states while making scheduling decisions, and the outcome of a transmission attempt during slot $t$ is known at the beginning of slot $t+1$. Let $b_n(t)\in\{0,1\}$ be the indicator function that represents the state of channel $n$ during slot $t$. The channel is ON, $b_n(t)=1$, with probability $\mu_n\in(0,1]$, and the channel is OFF, $b_n(t)=0$, with probability $1-\mu_n$. The channel state process is i.i.d.\ over time and independent across different channels. 


The \emph{reliability of channel} $n$ is represented by the probability of this channel being ON, $\mu_n$. Let $\vec\mu = \left[ \mu_1\ \cdots\ \mu_N \right]^\mathsf{T}$ be the vector of channel reliabilities. 
Let $\mu^*$ be the maximum channel reliability and let $n^*$ be the index of the corresponding channel, i.e.\,
	$\mu^* = \max_{n} \mu_n = \mu_{n^*}$. 
For simplicity, we assume that the optimal channel $n^*$ is unique. Naturally, if the channel reliabilities were known by the learning algorithm in advance, then the algorithm would select channel $n^*$ in every slot $t$. However, since the channel reliabilities $\vec\mu$ are initially unknown, the learning algorithm has to estimate $\mu_n$ using observations from previous transmission attempts, while at the same time attempting to minimize the AoI in the network. 
%
%
Next, we formulate the AoI minimization problem.

\subsection{Age of Information}
The AoI captures how old the information is from the perspective of the destination. Let $h_m(t)$ be a positive integer that represents the AoI associated with source $m$ at the beginning of slot $t$. By definition, we have $h_m(t):=t-\tau_m(t)$, where $\tau_m(t)$ is the generation time of the latest packet successfully transmitted from source $m$ to the destination\footnote{We define $\tau_m(t)=0$ prior to the first packet delivery from source $m$.}. If the destination does not receive a fresh packet from source $m$ during slot $t$, then in the next slot we have $h_m(t+1)=h_m(t)+1$, since the information at the destination is one slot older. In contrast, if the destination receives a fresh packet from source $m$ during slot $t$, then in the next slot the value of $\tau_m(t+1)$ is updated to the generation time of the received packet and the AoI is reduced by $\tau_m(t+1)-\tau_m(t)$. This difference is the ``freshness gain'' associated with the received packet. The evolution of $h_m(t)$ over time is illustrated in Fig.~\ref{fig.AoIevolution}. We define the vector of AoI in slot $t$ as $\vec{h}(t) = \left[ h_1(t)\ \cdots\ h_M(t) \right]^\mathsf{T}$.

\begin{figure}[t]
\begin{center}
\includegraphics[width=0.9\columnwidth]{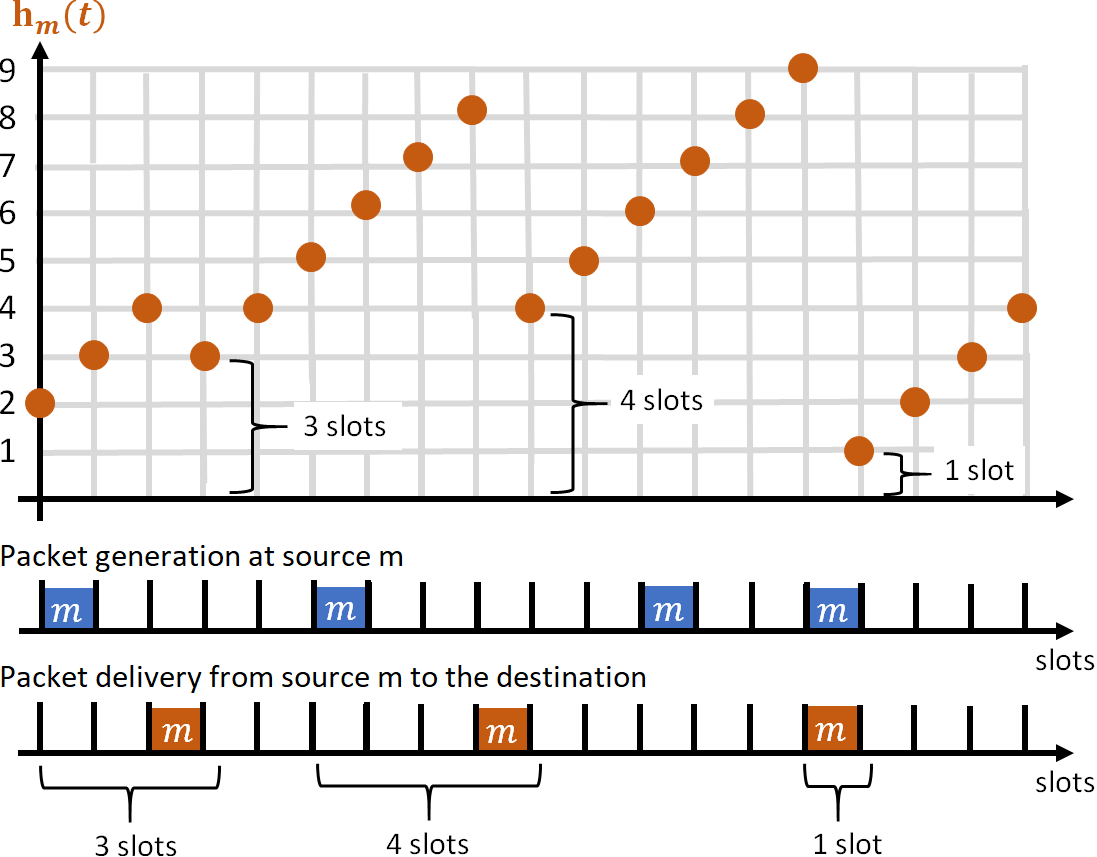}
\end{center}
\vspace{-0.3cm}
\caption{The blue and orange rectangles at the bottom represent packets generated at source $m$ and successful packet transmissions from source $m$, respectively. The orange curve shows the AoI evolution $h_m(t)$ associated with source $m$.} \label{fig.AoIevolution}
\vspace{-0.4cm}
\end{figure}

For capturing the information freshness of the entire network, we consider the \emph{expected total AoI} $\bar h(T)$, which is defined as the expected sum of the AoI over all sources and over time, namely
\begin{equation}\label{eq.totalAge}
	\bar h(T) = \mathrm{E}\left[\sum_{m=1}^{M} \sum_{t=1}^{T} h_m(t)\right] \; ,
\end{equation}
where the expectation is with respect to the randomness in the channel states $b_n(t)$, packet generation process $\vec a(t)$, and scheduling decisions $(m,n)$. 
The learning algorithm schedules pairs $(m,n)$ over time so as to minimize the expected total AoI $\bar h(T)$. Recall that in this sequential decision problem, the channel reliabilities $\mu_n$ are initially unknown by the learning algorithm and should be estimated over time. Next, we discuss the class of learning algorithms considered in this paper.


\subsection{Learning Algorithm}

In this section, we present three important concepts associated with the learning algorithm: the channel policy, the source policy, and the AoI regret. Prior to discussing these concepts, we introduce some notation. 
In each slot $t$, the learning algorithm selects a single source and a single channel. Let $m(t)$ be the index of the source selected during slot $t$ and let $n(t)$ be the index of the channel selected during slot $t$. Then, the pair selected in each slot can be denoted as $(m(t),n(t))$. Notice that the learning algorithm can be divided into two components: the source policy, which selects $m(t)$, and the channel policy, which selects $n(t)$. Let $b(t)=b_{n(t)}(t)$ be the state of the channel selected during slot $t$, and recall that $\vec{a}(t)$ is the vector of packet generations and $\vec{h}(t)$ is the vector of AoI in slot $t$. Using this notation, we define the \emph{channel policy} and the \emph{source policy}.

The \emph{channel policy} may (or may not) take into account the status of the transmission queues at the sources (in particular $E(t)$) in making scheduling decisions $n(t)$. Hence, we define two types of channel policies: queue-independent channel policies and queue-dependent channel policies. Let $\Pi_B$ be the class of admissible \emph{queue-independent channel policies} $\pi_b$. In slot $t$, an arbitrary policy $\pi_b\in\Pi_B$ selects $n(t)$ using information about the outcome of previous transmission attempts. In particular, the queue-independent channel history in slot $t$ is given by $H_B(t) = \{ n(1),b(1),\cdots,n(t-1),b(t-1) \}$. Let $\bar\Pi_B$ be the class of admissible \emph{queue-dependent channel policies} $\bar\pi_b$. In slot $t$, an arbitrary policy $\bar\pi_b\in\bar\Pi_B$ selects $n(t)$ using information about the outcome of previous transmission attempts and about the current status of the transmission queues. In particular, the queue-dependent channel history in slot $t$ is given by $\bar{H}_B(t) = H_B(t)\cup\{E(t)\}$. In Sec.~\ref{sec.Policy}, we show that this small amount of information, namely $E(t)$, can have a significant impact on the performance of the channel policy. 
It is easy to see that both the optimal queue-independent channel policy $\pi^*_b$ and the optimal queue-dependent channel policy $\bar \pi^*_b$ select the channel with highest reliability $\mu^*$ at every slot $t$. However, since the reliabilities $\vec\mu$ are not known a priori, the channel policies have to estimate $\vec\mu$ over time. In Sec.~\ref{sec.Regret}, we consider queue-independent channel policies and in Sec.~\ref{sec.Policy}, we consider queue-dependent channel policies. 


The \emph{source policies} considered in this paper are work-conserving, i.e.\, policies that never transmit dummy packets when there are undelivered data packets in the system. Let $\Pi_A$ be the class of admissible work-conserving source policies $\pi_a$. In slot $t$, an arbitrary source policy $\pi_a\in\Pi_A$ selects $m(t)$ using information about the current AoI and the generation times of the packets waiting to be transmitted at the sources' queues. In particular, the source history in slot $t$ is given by $H_A(t) = \{ \vec{a}(1),\vec{h}(1),\cdots,\vec{a}(t),\vec{h}(t)\}$. 
The optimal source policy $\pi^*_a\in\Pi_A$ is the transmission scheduling policy that minimizes the expected total AoI in \eqref{eq.totalAge}. A few works in the literature \cite{YuPin18,YuPinTMC,igorMobiHoc,igorTMC19} have addressed the problem of finding the transmission scheduling policy that minimizes AoI in wireless networks with stochastic packet generation and unreliable channels with \emph{known statistics}. Despite those efforts, a full characterization of the optimal source policy is still an open problem. 

In this paper, we consider learning algorithms $\pi$ that are a composition of a source policy and a channel policy $\pi=(\pi_a,\pi_b)$. Our goal is to study the exploration-exploitation dilemma faced by the channel policy. To that end, we analyze the AoI regret of learning algorithms employing the optimal source policy and different channel policies. To analyze the AoI regret of learning algorithms without the full characterization of the optimal source policy $\pi_a^*$, we derive lower and upper bounds on the regret.  
%
%
These bounds are discussed in Proposition~\ref{p1}, Proposition~\ref{t2}, and Theorem~\ref{t3}, where
we assumed that the optimal source policy $\pi_a^*$ 
is the same irrespective of the queue-independent channel policy $\pi_b$ under consideration, namely
\begin{equation}\label{eq.independenceAssumption}
	\pi^*_a = \argmin_{\pi_a\in\Pi_A} \mathrm{E}\left[ \sum_{m=1}^{M} \sum_{t=1}^{T} h^{(\pi_a,\pi_b)}_m(t) \right] \; , \; \forall\pi_b\in\Pi_B \; ,
\end{equation}
where $h^{(\pi_a,\pi_b)}_m(t)$ denotes the AoI associated with source $m$ in slot $t$ when the learning algorithm $\pi=(\pi_a,\pi_b)$ is employed. An analogous assumption is utilized for the case of queue-dependent channel policies $\bar\pi_b\in\bar\Pi_B$.

The \emph{AoI regret} of a learning algorithm $\pi$ with queue-independent channel policy $\pi_b$ is defined as the difference between the expected total AoI $\bar h^\pi (T)$ when $\pi=(\pi_a,\pi_b)$ is employed and the expected total AoI $\bar h^* (T)$ when the optimal algorithm $\pi^*=(\pi_a^*,\pi_b^*)$ is employed, namely
\begin{equation}\label{eq.AoIregret}
	R^\pi(T) = \mathrm{E}\left[ \sum_{m=1}^{M} \sum_{t=1}^{T} h^\pi_m(t) - \sum_{m=1}^{M} \sum_{t=1}^{T} h^*_m(t) \right] \; ,
\end{equation}
where the expectation is with respect to the randomness in the channel states $b(t)$, packet generation process $\vec a(t)$, and scheduling decisions $(m(t),n(t))$. The definition of AoI regret for a learning algorithm $\bar\pi$ with queue-dependent channel policy $\bar\pi_b$ is analogous to \eqref{eq.AoIregret}.
%
%
%
%
%
%
Next, we analyze the AoI regret of learning algorithms with \emph{queue-independent channel policies}.

\section{Regret Analysis}\label{sec.Regret}

The problem of learning channel reliabilities over time is closely related to the stochastic MAB problem. A natural class of channel policies to consider are traditional MAB algorithms such as $\epsilon$-Greedy, UCB, and TS.
%
%
In this section, we derive bounds on the AoI regret of learning algorithms that employ \emph{queue-independent channel policies}. Notice that the class of queue-independent channel policies $\Pi_B$ includes traditional MAB algorithms. We describe a learning algorithm employing TS as its channel policy in Algorithm~\ref{alg.TS}. 

\begin{algorithm}[t]
	\SetAlgoLined
	Initialization: time $t=1$, estimates $\hat\mu_n=0$, counters $T_n=0$, parameters $\alpha_n=\beta_n=1$, $\forall n\in\{1,\cdots,N\}$\;
	\While{$1\le t\le T$}{
		Optimal source policy selects $m\in\{1,2,\cdots,M\}$\;
		$\theta_n \sim \text{Beta}(\alpha_n,\beta_n)$\;
		$n=\argmax_{n'\in\{1,\cdots,N\}}\theta_{n'}$\;
		Source $m$ transmits packet through channel $n$ and observes channel state $b$\;
		\eIf{$b=1$}{
			$\alpha_n = \alpha_n + 1$\;
		}{
			$\beta_n = \beta_n + 1$\;
		}
		Compute new estimate $\hat\mu_n = \dfrac{\hat\mu_n T_n + b}{T_n+1}$\;
		$T_n=T_n+1$\;
		$t=t+1$\;
	}
	\caption{Learning Algorithm employing TS as its channel policy}\label{alg.TS}
\end{algorithm}


Scheduling decisions of a learning algorithm $\pi$ might differ from those of $\pi^*$ both in the source and in the channel, which makes the analysis of the AoI regret $\sum_{t=1}^T \sum_{m=1}^M \mathrm{E}[ h_m^\pi(t)-h_m^{*}(t) ]$ challenging. To alleviate this challenge, we use stochastic coupling to create \emph{equivalent coupled channel state processes} that are simpler to analyze. Similar coupling arguments were employed in \cite{RegretQueueing,RegretAoI}.
\begin{remark}[Coupled Channel States]\label{r1}
	Let $\{U(t)\}_{t=1}^T$ be a sequence of i.i.d.\ random variables uniformly distributed in the interval $[0,1]$. In each slot $t$, the channel states $b_n(t)$ are determined as follows
	\begin{equation}\label{coupling}
		b_n(t) = 1 \iff 0\le U(t) \le \mu_n \; , \; \forall n \; .
	\end{equation}
\end{remark}
By construction, the coupled channel states are no longer independent. In particular, if a channel is ON during slot $t$, then all channels with higher reliability $\mu_n$ are also ON during that slot. Notice that, in each slot $t$, each coupled channel $n$ has the same probability distribution as the associated original channel~$n$, namely $\mathrm{P}\left(b_n(t)=1\right)=\mu_n,\forall n,t$. 
Hence, given the scheduling decision $(m(t),n(t))$ of $\pi$ during any slot $t$, the probability of a successful transmission attempt from source $m(t)$ through channel $n(t)$ is the same for both the coupled and original channel states. It follows that 
%
%
the probability distribution of $h_m^\pi(t)$ also remains the same for all slots $t$ and for all sources $m$ and, thus, the AoI regret $R^\pi(T)$ in \eqref{eq.AoIregret} also remains the same for both the coupled and original channel state processes.
%
\emph{For simplicity of analysis, henceforth in this paper, we assume that the channel state processes are coupled as described in Remark~\ref{r1}}.

In Proposition~\ref{p1}, Proposition~\ref{t2}, and Corollary~\ref{c3}, we derive bounds on the AoI regret of a learning algorithm $\pi$ with respect to its \emph{expected number of suboptimal channel choices}, namely
\begin{equation}\label{eq.suboptimalChoices}
	\mathrm{E}[K^\pi(T)] = \mathrm{E}\left[\sum_{t=1}^T\mathbbm{1}\left\{n^\pi(t) \neq n^*\right\}\right] \; ,
\end{equation}
where $\mathbbm{1}\left\{n^\pi(t) \neq n^*\right\}=1$ if $n^\pi(t) \neq n^*$, and $\mathbbm{1}\left\{n^\pi(t) \neq n^*\right\}=0$ otherwise. 
%
We consider two classes of admissible learning algorithms
\begin{align}
\Pi = \left\{ \pi=(\pi_a,\pi_b) : \pi_a\in\Pi_A , \pi_b\in\Pi_B \right\} \; ; \label{eq.Pi}\\
\Pi^* = \left\{ \pi=(\pi_a,\pi_b) : \pi_a = \pi^*_a , \pi_b\in\Pi_B \right\} \; . \label{eq.Pi_star}
\end{align}
Both classes employ queue-independent channel policies. The difference is that $\Pi$ employs any admissible source policy $\pi_a\in\Pi_A$, while $\Pi^*$ employs the optimal source policy $\pi_a^*$. Naturally, we have $\Pi^*\subset\Pi$. 

\begin{proposition}[Lower Bound]\label{p1}
	For any given network configuration $(\lambda,\vec\mu)$, the AoI regret of any learning algorithm $\pi\in\Pi$ scales at least on the order of its expected number of suboptimal channel choices, namely\footnote{$f(t) = \Omega(g(t)) \iff \exists C>0\ \exists t_0\ \forall t>t_0 : f(t)\ge C\cdot g(n)$}
	\begin{equation}\label{eq.p1}
		R^\pi(T) = \Omega\left( \mathrm{E}\left[ K^{\pi}(T) \right] \right) \; .
	\end{equation}
\end{proposition}
    
    
    \noindent\emph{Proof outline.} In addition to the suboptimal channel choices, source choices $m^\pi(t)$ of algorithm $\pi\in\Pi$ can also differ from the source choices $m^{*}(t)$ of $\pi^*$. To overcome this challenge, we construct an auxiliary algorithm $\hat\pi^*$ with optimal channel policy 
    and a source policy that selects the same source\footnote{Notice that if the selected source $m^{\pi}(t)$ has no packet in its transmission queue, then the auxiliary algorithm attempts to transmit a dummy packet.} $m^{\pi}(t)$ as $\pi$ in every slot $t$. Then, we focus on the \emph{auxiliary AoI regret} $\sum_{t=1}^T \sum_{m=1}^M \mathrm{E}[ h_m^\pi(t)-h_m^{\hat\pi^*}(t)]$ associated with the auxiliary algorithm $\hat\pi^*$, which we show to be not greater than the original AoI regret $\sum_{t=1}^T \sum_{m=1}^M \mathrm{E}[ h_m^\pi(t)-h_m^{*}(t) ]$. We then observe that each suboptimal channel choice of $\pi$ results in a penalty to the auxiliary AoI regret, and we show that this penalty is lower bounded by a constant. Using this constant, we obtain the desired lower bound on the original AoI regret in \eqref{eq.p1}. 
    The details are omitted due to the space constraint.
    %
    %
    %

\begin{proposition}[Upper Bound]\label{t2}
    For any given network configuration $(\lambda,\vec\mu)$, the AoI regret of any learning algorithm $\pi\in\Pi^*$ scales at most on the order of its expected number of suboptimal channel choices, namely\footnote{$f(t) = O(g(t)) \iff \exists C>0\ \exists t_0\ \forall t>t_0 : f(t)\le C\cdot g(n)$}
	\begin{equation}\label{eq.t2}
		R^\pi(T) = O\left( \mathrm{E}\left[ K^{\pi}(T) \right] \right) \; .
	\end{equation}
\end{proposition}
    \noindent\emph{Proof outline}. 
    Despite the fact that both learning algorithms $\pi\in\Pi^*$ and $\pi^*$ employ the same optimal source policy $\pi_a^*$, they might select different sources $m^{\pi}(t) \neq m^{*}(t)$ over time, due to their different channel policies. 
    To address this challenge, we use an approach similar to the proof of Proposition~\ref{p1}. We construct an auxiliary algorithm $\hat\pi\in\Pi^*$ 
    with a source policy that selects the same source $m^{*}(t)$ as $\pi^*$ in every slot $t$, and with a channel policy that selects the same channel $n^{\pi}(t)$ as $\pi$ in every slot $t$. Then, we show that the auxiliary AoI regret $\sum_{t=1}^T \sum_{m=1}^M \mathrm{E}[ h_m^{\hat\pi}(t) - h_m^{*}(t) ]$ associated with the auxiliary algorithm $\hat\pi$ is not lower than the original AoI regret $\sum_{t=1}^T \sum_{m=1}^M \mathrm{E}[ h_m^\pi(t)-h_m^{*}(t) ]$. To derive an upper bound on the auxiliary AoI regret, we analyze the penalty that results from each suboptimal channel choice of $\hat\pi$. 
    During a slot $t$ where $\hat\pi$ makes a suboptimal channel choice, if channel $n^{\hat\pi}(t)$ is OFF and channel $n^{*}$ is ON, then a \emph{discrepancy} is added to the difference between the AoI of $\hat\pi$ and the AoI of $\pi^*$, i.e.\, $h_m^{\hat\pi}(t+1) - h_m^{*}(t+1)>h_m^{\hat\pi}(t) - h_m^{*}(t)$.
    This discrepancy lasts until the next successful transmission of a packet from source $m$ by the auxiliary algorithm $\hat\pi$, after which the values of $h_m^{\hat\pi}(\cdot)$ and $h_m^{*}(\cdot)$ become equal\footnote{Recall from Remark~\ref{r1} that channel states are coupled. Hence, if channel $n^{\hat\pi}(t)$ is ON, then channel $n^{*}$ is also ON.}. We refer to the duration of the discrepancy as its \emph{length}. The penalty that results from a suboptimal channel choice is the product of the discrepancy and its length. We characterize the auxiliary AoI regret by expressing it as the sum of the penalties arising from suboptimal channel choices. Then, using discrete phase-type distributions, we upper bound the discrepancies and the lengths by constants (in the expected sense) to obtain the result in \eqref{eq.t2}. 
    The details are omitted due to the space constraint.

\begin{corollary}\label{c3}
For any given network configuration $(\lambda,\vec\mu)$, the AoI regret of any learning algorithm $\pi\in\Pi^*$ scales with its expected number of suboptimal channel choices, namely\footnote{$f(t) = \Theta(g(t)) \iff f(t) = O(g(t)) \land f(t) = \Omega(g(t)) \iff \exists C_1,C_2>0\ \exists t_0\ \forall t>t_0 : C_1\cdot g(n) \le f(t)\le C_2\cdot g(n)$} 
\begin{equation}
R^\pi(T) = \Theta\left( \mathrm{E}\left[ K^{\pi}(T) \right] \right) \; . 
\end{equation}
\end{corollary}

Corollary~\ref{c3} follows directly from Propositions~\ref{p1} and ~\ref{t2}. Notice that the bounds in Proposition~\ref{t2} and Corollary~\ref{c3} are not valid for the broader class of learning algorithms $\Pi$ which includes suboptimal source policies. This is because suboptimal source choices may add to the AoI regret, possibly making it grow faster than $\mathrm{E}\left[ K^{\pi}(T) \right]$.

Prior to analyzing the AoI regret of learning algorithms that employ $\epsilon$-Greedy, UCB, and TS as their channel policy, we define $\alpha$\emph{-consistent learning algorithms} \cite{surveyMAB3,RegretQueueing} and discuss a few of their properties. Let $\mathrm{E}[T^\pi_n(T)]$ be the expected number of times channel $n$ is selected by $\pi\in\Pi$ in the first $T$ slots, namely
\begin{equation}\label{eq.channelChoices}
	\mathrm{E}[T_n^\pi(T)] = \mathrm{E}\left[\sum_{t=1}^T\mathbbm{1}\left\{n^\pi(t) = n\right\}\right] \; .
\end{equation}
\begin{definition}[$\alpha$-consistency]\label{def.alphaConsistency}
	For a given $\alpha\in(0,1)$, a learning algorithm $\pi\in\Pi$ is classified as \emph{$\alpha$-consistent} if, for any network configuration $(\lambda,\vec\mu)$, we have
		$E\left[ T^\pi_n(T) \right] = O(T^\alpha)$  
	for \emph{all suboptimal channels} $n\neq n^*$.	
\end{definition}
Intuitively, a learning algorithm $\pi\in\Pi$ is $\alpha$-consistent if its channel policy has good performance in \emph{every network configuration}. Consider a learning algorithm with a trivial channel policy that selects $n(t)=1$ in every slot $t$. In network configurations with $n^*=1$, this channel policy never selects suboptimal channels, i.e. $\mathrm{E}\left[ T^\pi_n(T) \right] = O(T^\alpha),\forall n\neq n^*$. However, in network settings with $n^*\neq 1$, this channel policy is such that $\mathrm{E}\left[ T^\pi_1(T) \right] = T$, which violates the definition of $\alpha$-consistency. In the remainder of this section, we focus on channel policies that have good performance in every network configuration. In particular, we analyze the AoI regret of $\alpha$-consistent learning algorithms with queue-independent channel policies. 

\begin{remark}[AoI regret of $\alpha$-consistent algorithms]\label{r2}
    In \cite[Corollary~$20$]{RegretQueueing}, the authors show that any learning algorithm $\pi\in\Pi$ that is $\alpha$-consistent has an expected number of suboptimal channel choices that scales as
	    $E\left[ K^{\pi}(T) \right] = \Omega(\log T)$,
    for any network configuration $(\lambda,\vec\mu)$. Hence, it follows from the lower bound in Proposition~\ref{p1} that the associated AoI regret scales as
    \begin{equation}\label{eq.alphaOmega}
        R^\pi(T) = \Omega(\log T) \; ,
    \end{equation}
    for any network configuration $(\lambda,\vec\mu)$.
\end{remark}
Notice that the lower bound in Remark~\ref{r2} applies to $\alpha$-consistent learning algorithms with queue-independent channel policies that do not know the statistics of the channels in advance. 

Learning algorithms that employ $\epsilon$-Greedy, UCB, and TS as their channel policy are known to have suboptimal channel choices scaling as $\mathrm{E}[K^\pi(T)] = O(\log T)$ for any network configuration $(\lambda,\vec\mu)$ \cite{FiniteMAB,furtherTS}, which implies that they are $\alpha$-consistent. Hence, it follows from the upper bound in Proposition~\ref{t2} and from \eqref{eq.alphaOmega} that the AoI regret of these learning algorithms scale as 
\begin{equation}
R^\pi(T) = \Theta(\log T) \; .
\end{equation}

In \cite{RegretAoI}, the authors derived lower and upper bounds on the AoI regret of learning algorithms employing queue-independent channel policies, including UCB and TS, in networks with a single source generating and transmitting fresh packets in every slot $t$. Propositions~\ref{p1} and \ref{t2} generalize the results in \cite{RegretAoI} to networks with multiple sources generating packets according to stochastic processes. The analysis of the AoI regret is more challenging in this network setting for the following reasons: 
i) the optimal source policy $\pi^*_a$ is unknown and there is no closed-form expression for the expected total AoI \eqref{eq.totalAge} of the optimal algorithm $\pi^*=(\pi_a^*,\pi_b^*)$; and 
ii) the learning algorithm under consideration $\pi=(\pi_a,\pi_b)$ can make suboptimal choices both in terms of sources $m(t)$ and channels $n(t)$, and these two types of suboptimal choices affect the AoI regret $R^\pi(T)$ differently. 
Next, we develop a learning algorithm that leverages information about the status of the transmission queues in making scheduling decisions $n(t)$, and show that this new learning algorithm has $O(1)$ AoI regret.

\section{Order-Optimal Learning Algorithm}\label{sec.Policy}
In this section, we develop a learning algorithm $\bar\eta\in\bar\Pi$ with a \emph{queue-dependent channel policy} that selects $n(t)$ using information about the outcome of previous transmission attempts, namely
$H_B(t) = \{ n(1),b(1),\cdots,n(t-1),b(t-1) \}$,
and about the current status of the transmission queues, $E(t)$. Then, we derive an upper bound on its AoI regret. In particular, we show that the AoI regret of $\bar\eta$ is such that $R^{\bar\eta}(T)=O(1)$. Notice that the only difference between the learning algorithms $\pi\in\Pi$ in Sec.~\ref{sec.Regret} and the order-optimal learning algorithm $\bar\eta$ is the knowledge of $E(t)$. This seemingly modest addition led to the reduction of the AoI regret from $R^{\pi}(T)=\Omega(\log T)$ to $R^{\bar\eta}(T)=O(1)$. To the best of our knowledge, this is the first learning algorithm with bounded AoI regret.

The key insight is that when packets are generated randomly, the learning algorithm $\bar\eta$ can utilize times when the network has no data packets to transmit, i.e.\ when $E(t)=1$, to transmit dummy packets and learn the statistics of the channels without incurring an opportunity cost. The order-optimal learning algorithm $\bar\eta=(\eta_a,\bar\eta_b)$ has optimal source policy $\eta_a=\pi^*_a$ and a channel policy $\bar\eta_b\in\bar\Pi_B$ that operates as follows: when the system is empty, $E(t)=1$, the policy chooses a channel uniformly at random and uses the outcome of the transmission attempt to update its estimates of the channel reliabilities and, when the system is nonempty, $E(t)=0$, the policy chooses the channel with the current highest estimated reliability. Notice that the channel policy only updates its estimates of the channel reliabilities when the system is empty. A similar channel policy was used in \cite{RegretThomas,ThomasPHD} to develop a learning algorithm with bounded queue-length regret. The order-optimal learning algorithm $\bar\eta$ is described in Algorithm~\ref{alg.Opt}. The upper bound on the AoI regret is established in the theorem that follows.

\begin{algorithm}[t]
	\SetAlgoLined
	Initialization: time $t=1$, estimates $\hat\mu_n=0$, counters $T_n=0$, $\forall n\in\{1,\cdots,N\}$\;
	\While{$1\le t\le T$}{
		Optimal source policy selects $m\in\{1,2,\cdots,M\}$\;
		\eIf{system is empty}{
			$n = \text{Unif}\{1,\cdots,N\}$\;
			Source $m$ transmits dummy packet through channel $n$ and observes channel state $b$\;
			$\hat\mu_{n} = \dfrac{\hat\mu_{n} T_{n} + b}{T_{n}+1}$\;
			$T_{n}=T_{n}+1$\;
		}{
			$n=\argmax_{n'\in\{1,\cdots,N\}}\hat\mu_{n'}$\;
			Source $m$ transmits data packet through channel $n$ and observes channel state $b$\;
		}
		$t=t+1$\;
	}
	\caption{Order-Optimal Learning Algorithm}\label{alg.Opt}
\end{algorithm}

\begin{theorem}\label{t3}
For any given network configuration $(\lambda,\vec\mu)$, the AoI regret of the order-optimal learning algorithm $\bar\eta$ is bounded, namely
\begin{equation}\label{eq.t3}
	R^{\bar\eta}(T)=O(1) \; .
\end{equation}
\end{theorem}
\begin{proof}
	Recall that the two components of the order-optimal learning algorithm are $\bar\eta = (\eta_a,\bar\eta_b)$. Similarly to the proof of Proposition~\ref{t2}, we start by constructing an auxiliary algorithm $\hat\eta = (\hat\eta_a,\bar\eta_b)$ which has a source policy $\hat\eta_a$ that selects the same source $m^{*}(t)$ as $\pi^*$ in every slot $t$. Since $\eta_a = \pi^*_a$ is optimal, it follows that $\hat\eta_a$ is suboptimal, which implies that
	\begin{align}
		\sum_{m=1}^M \sum_{t=1}^T \mathrm{E}[ h_m^{\bar\eta}(t)-h_m^{*}(t) ] \le \sum_{m=1}^M \sum_{t=1}^T \mathrm{E}[ h_m^{\hat\eta}(t)-h_m^{*}(t) ] \; . \label{aux_regret}
	\end{align}
	We denote the RHS of \eqref{aux_regret} as the \emph{auxiliary AoI regret}. Prior to deriving the upper bound on the auxiliary AoI regret, we introduce some definitions that are particular to the channel policy $\bar\eta_b$.
	
	Consider the time slots when the system becomes empty, i.e.\ time slots $t$ such that $E(t-1)=0$ and $E(t)=1$. We denote the time interval between two such slots as a \emph{period} 
	and we divide time $t\in\{1,2,\cdots,T\}$ into successive periods, with period index $p\in\{1,2,\cdots,P\}$. By definition, the system is empty, $E(t)=1$, in the beginning of each period $p$ and it remains empty until the first packet generation. Once the first packet is generated, the system becomes nonempty, $E(t)=0$, and it remains nonempty until the end of the period. Hence, each period $p$ has two phases: an \emph{empty phase} and a \emph{nonempty phase}, with each phase having at least one slot.
	Let $s_p$ and $f_p$ be the first and the last slots of period $p$, respectively, with $s_1=1$ and $s_{p+1} = f_p + 1,\forall p$. 
	Then, the cumulative AoI of source $m$ during period $p$ can be written as
	\begin{equation}\label{eq.cumy}
		y^{\hat\eta}_m(p) = \sum_{t=s_p}^{f_p} h^{\hat\eta}_m(t) \; .
	\end{equation}
	
	Recall from Algorithm~\ref{alg.Opt} that estimates of the channel reliabilities are only updated during empty phases. Within a nonempty phase, the estimates do not change and, thus, the selected channel also does not change. Let $\bar n(p)$ be the channel selected by policy $\bar\eta_b$ during the entire \emph{nonempty phase} of period $p$. 
	If $\bar n(p)=n^*$, we refer to period $p$ as an \emph{optimal period}. Otherwise, we refer to period $p$ as a \emph{suboptimal period}.
	Next, we derive an upper bound on the auxiliary AoI regret in terms of the expected AoI contributions of the suboptimal periods.
	
	\begin{lemma}\label{step1}
		The auxiliary AoI regret is upper bounded by
		\begin{align}
			&\sum_{m=1}^M \sum_{t=1}^T \mathrm{E}[ h_m^{\hat\eta}(t)-h_m^{*}(t)] \nonumber\\
			&\quad \le \sum_{m=1}^{M} \sum_{p=1}^{T} \mathrm{E}\left[ y^{\hat\eta}_m(p) \;\middle\vert\; \bar n(p)\neq n^* \right] \mathrm{P}\left(\bar n(p)\neq n^*\right) \; .\label{eqstep1}
		\end{align}
	\end{lemma}
	To establish Lemma~\ref{step1}, we first show that if period $p$ is an optimal period, then $h^{\hat\eta}_m(t) = h^{*}_m(t), \forall m, \forall t \in \{s_p,\cdots,f_p\}$, which implies that optimal periods do not contribute to the auxiliary AoI regret. Then, we obtain the upper bound in \eqref{eqstep1} by manipulating the expression of the auxiliary AoI regret.
	The complete proof of Lemma~\ref{step1} can be found in Appendix \ref{newA}.
	
	In Lemmas~\ref{step2} and \ref{step3}, we derive upper bounds on the first and second terms on the RHS of \eqref{eqstep1}, respectively.
	\begin{lemma}\label{step2}
		There exists a constant $C_y$ such that
		\begin{equation}\label{eqstep2}
			\mathrm{E}\left[ y^{\hat\eta}_m(p) \;\middle\vert\; \bar n(p)\neq n^* \right] \le C_y \; .
		\end{equation}
	\end{lemma}
	To establish Lemma~\ref{step2}, we first show that the cumulative AoI $y^{\hat\eta}_m(p)$ of source $m$ in period $p$ can be upper bounded by
	\begin{align}
		y^{\hat\eta}_m(p) &= \sum_{t=s_p}^{f_p} h^{\hat\eta}_m(t) \le \sum_{i=0}^{f_p-s_p} \left( h_m^{\hat\eta}(s_p) + i \right)\nonumber\\
		&= h_m^{\hat\eta}(s_p) [f_p - s_p + 1] + \dfrac12 [(f_p - s_p)^2 + f_p - s_p] \; .\label{cumsum}
	\end{align}
	Then, we derive an upper bound on the conditional expectation of \eqref{cumsum}. In particular, we
	show that $h_m^{\hat\eta}(s_p)$ can be upper bounded by a geometric random variable. Then, we show that the random variable $f_p - s_p$, which represents the length of period $p$, follows a discrete phase-type distribution. The upper bound on the conditional expectation of \eqref{cumsum} follows from the fact that the geometric random variable has finite second moment and the phase-type random variable has finite first and second moments. 
	The details are omitted due to the space constraint.
	
	%
	
	\begin{lemma}\label{step3}
		There exists a constant $C_p$ such that
		\begin{equation}
			\sum_{p=1}^{T} \mathrm{P}\left(\bar n(p)\neq n^*\right) \le C_p \; .\label{eqstep3}
		\end{equation}
	\end{lemma}
	To establish Lemma~\ref{step3}, we use Hoeffding's inequality to upper bound $\mathrm{P}\left(\bar n(p) = n\right)$ by an exponential function of $-p$, for every suboptimal channel $n$. The result in \eqref{eqstep3} follows directly from this upper bound. The complete proof of Lemma~\ref{step3} can be found in Appendix \ref{newC}.
	
	From the upper bound on the AoI regret in \eqref{aux_regret} and the results in Lemmas~\ref{step1}, \ref{step2} and \ref{step3}, we have
	\begin{align}
		R^{\bar\eta}(T) &\le \sum_{m=1}^M \sum_{t=1}^T \mathrm{E}[ h_m^{\hat\eta}(t)-h_m^{*}(t) ]\nonumber\\
		&\le \sum_{m=1}^{M} \sum_{p=1}^{T} \mathrm{E}\left[ y^{\hat\eta}_m(p) \;\middle\vert\; \bar n(p)\neq n^* \right] \mathrm{P}\left(\bar n(p)\neq n^*\right)\nonumber\\
		&\le C_y M C_p
	\end{align}
	which establishes the bound in \eqref{eq.t3}.
\end{proof}

In the particular case of a network with sources generating fresh packets at every slot $t$, i.e.\ $\lambda=1$, the algorithm $\bar\eta$ cannot utilize slots in which the system is empty to learn the channel reliabilities without incurring a cost in terms of AoI regret, which results in a $R^{\bar\eta}(T)$ that grows over time. 
The upper bound in Theorem~\ref{t3} is only valid for the network models described in Sec.~\ref{sec.Model}, in which $\lambda\in(0,1)$. 
Next, we evaluate the AoI regret of the different learning algorithms discussed in this paper using MATLAB simulations and we propose a heuristic algorithm that leverages the fast learning rates of TS and the bounded regret of the order-optimal algorithm.





\section{Simulations}\label{sec.Simulations}

In this section, we evaluate the performance of learning algorithms in terms of the AoI regret in \eqref{eq.AoIregret}. We compare learning algorithms employing the Age-Based Max-Weight source policy \cite[Sec.~5]{igorTMC19} and different channel policies, namely: i) $\epsilon$-Greedy; ii) UCB; iii) TS; iv) Optimal; and v) Hybrid. The Age-Based Max-Weight source policy selects, in each slot $t$, the source $m$ associated with the packet that gives the largest AoI reduction, $\tau_m(t+1)-\tau_m(t)$, if the transmission in slot $t$ is successful. Intuitively, this policy is selecting the source with highest potential reward in terms of AoI. In \cite{igorTMC19}, the authors evaluate the performance of the Age-Based Max-Weight source policy both analytically and using simulations, and show that it achieves near optimal AoI. The first three channel policies, namely $\epsilon$-Greedy, UCB, and TS, were discussed in Sec.~\ref{sec.Regret}. The Optimal policy is the order-optimal channel policy $\bar\eta_b$ developed in Sec.~\ref{sec.Policy}. The Hybrid policy employs TS for a fixed period in the beginning of the simulation and then employs the Optimal policy in the remaining slots.


We simulate a network with a time horizon of $T=10^5$ slots, $M=3$ sources, each generating packets according to a Bernoulli process with rate $\lambda$, and $N=5$ channels with reliabilities $\vec\mu = [0.4\ 0.45\ 0.5\ 0.55\ 0.6]^\mathsf{T}$. Figures~\ref{fig.Regret010} and \ref{fig.Regret075} show simulation results of the evolution of the AoI regret over time for $\lambda=0.1$ and $\lambda=0.75$, respectively. 
Figure~\ref{fig.Estimates} shows simulation results of the evolution of the reliability estimates associated with the channels with $\mu_4=0.55$ and $\mu_5=0.6$ over time for $\lambda=0.75$. 
Each data point in Figs.~\ref{fig.Regret010}, \ref{fig.Regret075}, and \ref{fig.Estimates} is an average over the results of $10^3$ simulations.


\begin{figure}[t]
\begin{center}
\includegraphics[width=0.9\columnwidth]{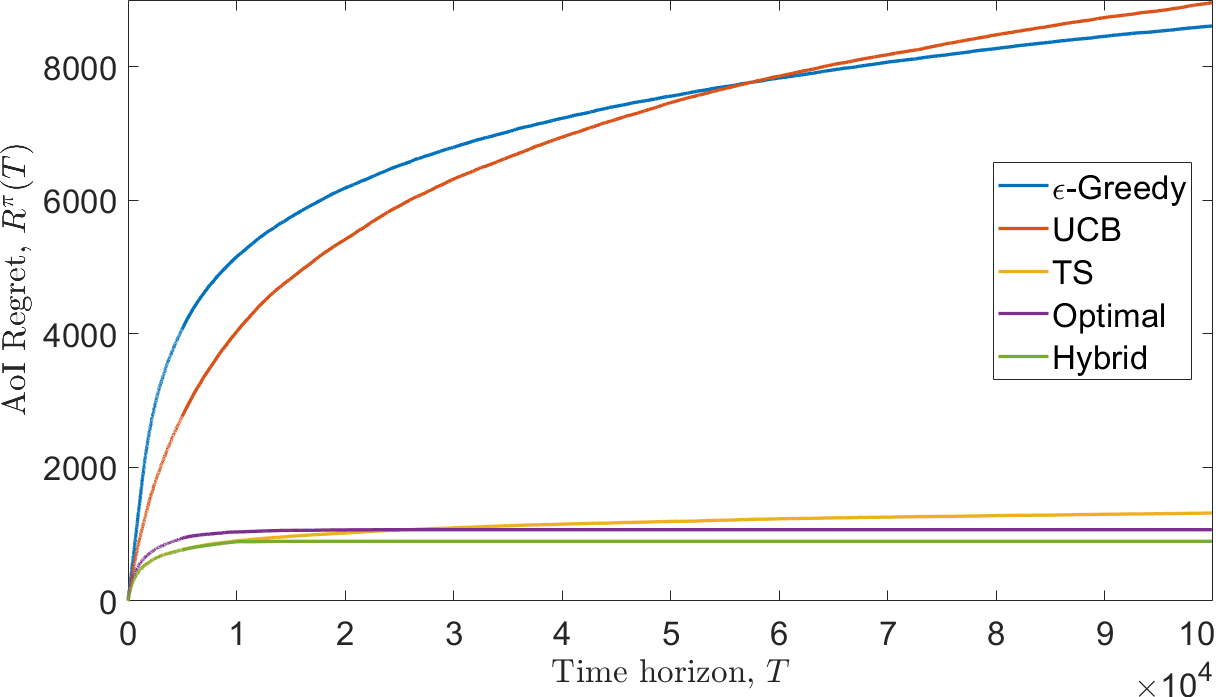}
\end{center}
\vspace{-0.4cm}
\caption{Simulation of a network with $\lambda=0.1$.} \label{fig.Regret010}
\vspace{-0.2cm}
\end{figure}

\begin{figure}[t]
\begin{center}
\includegraphics[width=0.9\columnwidth]{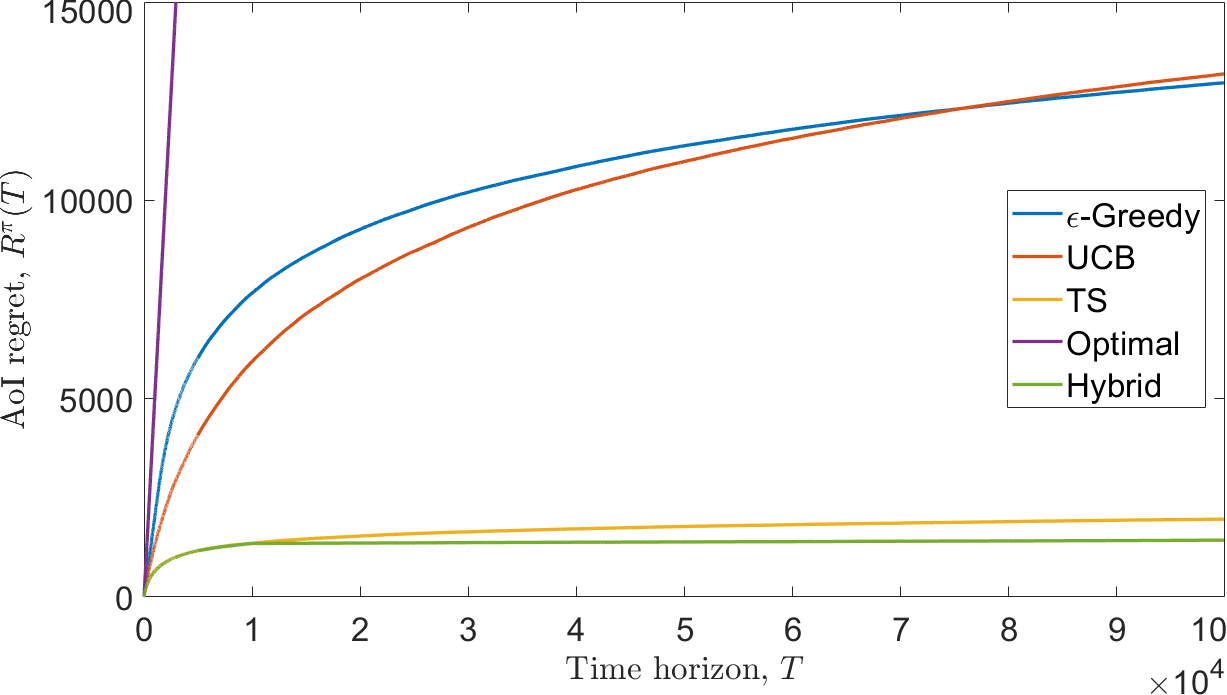}
\end{center}
\vspace{-0.4cm}
\caption{Simulation of a network with $\lambda=0.75$.} \label{fig.Regret075}
\vspace{-0.2cm}
\end{figure}

\begin{figure}[t]
\begin{center}
\includegraphics[width=0.9\columnwidth]{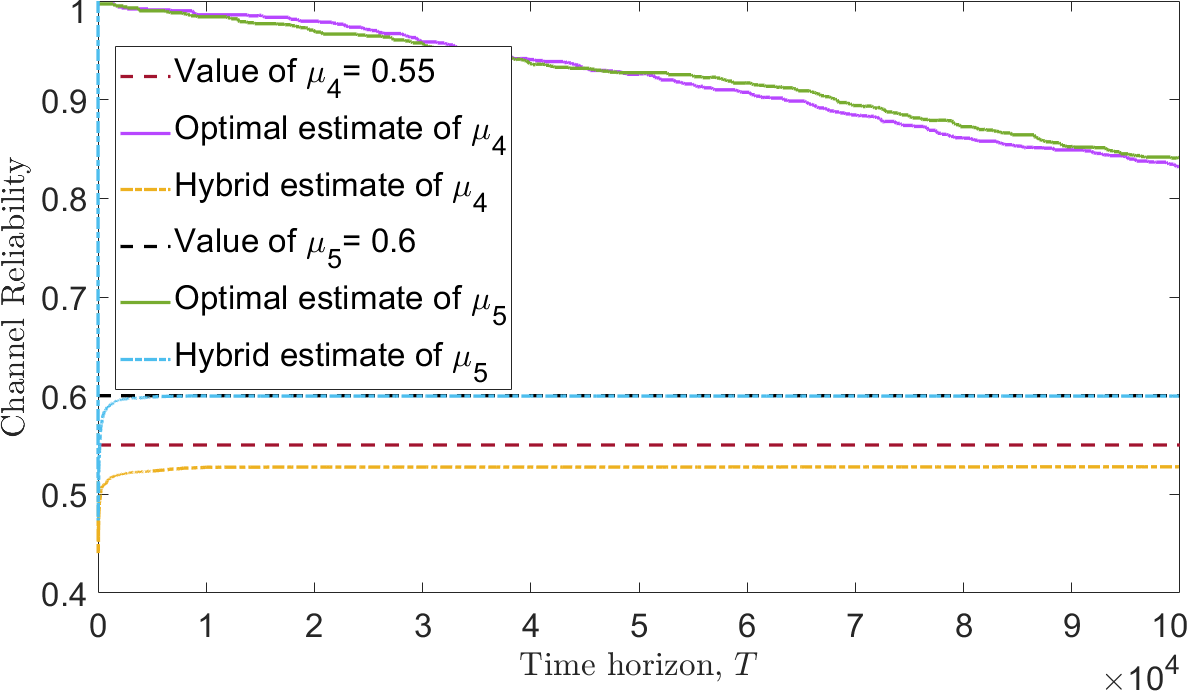}
\end{center}
\vspace{-0.4cm}
\caption{Simulation of a network with $\lambda=0.75$.} \label{fig.Estimates}
\vspace{-0.2cm}
\end{figure}


The results in Figs.~\ref{fig.Regret010} and \ref{fig.Regret075} suggest that, as expected, the AoI regret associated with Optimal and Hybrid is bounded, while the AoI regrets associated with $\epsilon$-Greedy, UCB and TS grow over time. By comparing the AoI regret of Optimal and TS in Figs.~\ref{fig.Regret010} and \ref{fig.Regret075}, it is clear that the AoI regret of the Optimal channel policy varies significantly with $\lambda$. In particular, for $T=10^5$, when $\lambda$ increases from $0.1$ to $0.75$, the AoI regret of TS increases by a factor of $1.5$ (from $1,318$ to $1,963$), while the AoI regret of Optimal increases by a factor of $491.0$ (from $1,068$ to $481,700$). 
%
%
%
%
A main reason for this performance degradation is that when $\lambda$ increases, empty systems with $E(t)=1$ occur less often and, as a result, the Optimal channel policy takes longer to learn the reliability of the channels, as can be seen in Fig.~\ref{fig.Estimates}. 
To improve the performance of the Optimal policy for networks with large $\lambda$, we propose a heuristic policy called Hybrid channel policy, which employs TS in the first $10^4$ slots to quickly learn the reliability of the channels, and then shifts to the Optimal policy which has bounded AoI regret in the long term. Figure~\ref{fig.Estimates} illustrates the difference in the learning rates between Optimal and Hybrid. 
Notice in Fig.~\ref{fig.Estimates} that there are extended periods of time in which the Optimal channel policy assigns a larger estimated reliability to a suboptimal channel, which leads to the large AoI regret shown in Fig.~\ref{fig.Regret075}. However, as established in Theorem~\ref{t3}, for a long enough time-horizon $T$, the Optimal policy will eventually converge to the true reliabilities, at which point the AoI regret will stop increasing.

\section{Conclusion}\label{sec.Conclusion}

This paper considers a single-hop wireless network with $M$ sources transmitting time-sensitive information to the destination over $N$ unreliable channels. Packets from each source are generated according to a Bernoulli process with known rate $\lambda$ and the state of channel $n$ (ON/OFF) varies according to a Bernoulli process with unknown rate $\mu_n$. The reliabilities $\vec\mu$ of the wireless channels is to be learned through observation. At every slot $t$, the learning algorithm selects a single pair $(m(t), n(t))$ and the selected source $m(t)$ attempts to transmit its packet via the selected channel $n(t)$. The goal of the learning algorithm is to minimize the expected total AoI $\bar h(T)$. To analyze the performance of the learning algorithm, we derive bounds on the AoI regret $R^\pi(T)$ associated with different learning algorithms. 
Our main contributions include: i) analyzing the performance of learning algorithms that employ channel policies based on traditional MAB algorithms ($\epsilon$-Greedy, UCB, and TS) and showing that their AoI regret scales as $\Theta(\log T)$; and ii) developing a novel learning algorithm and establishing that it has $O(1)$ AoI regret. To the best of our knowledge, this is the first learning algorithm with bounded AoI regret. Interesting extensions of this work include consideration of sources with unknown packet generation rates and channels with time-varying statistics.


\bibliographystyle{IEEEtran}
\bibliography{references}

\appendices

\section{Proof of Lemma~\ref{step1}}\label{newA}

To establish Lemma~\ref{step1}, we first show that \emph{optimal periods} do not contribute to the auxiliary AoI regret defined in \eqref{aux_regret}. Then, we obtain an upper bound on the contribution of \emph{suboptimal periods} by manipulating the expression of the auxiliary AoI regret.

\begin{lemma}\label{newl1}
	If period $p$ is an optimal period, then for any slot~$t$ within period~$p$ and for any source $m$, we have $h^{\hat\eta}_m(t) = h^{*}_m(t)$.
\end{lemma}
\begin{proof}
	Consider the time slots preceding period $p$ in a network employing the auxiliary learning algorithm $\hat\eta$. In slot $s_p-1$, the algorithm $\hat\eta$ delivers the last packet in the system and in slot $s_p$ the system becomes empty, with $E(s_p)=1$. Let $t'_m$ be the slot in  which the latest packet generated from source $m$ was delivered by $\hat\eta$. It follows that $t'_m \leq s_p-1$ and source $m$ did not generate new packets during the time interval $[t'_m+1,s_p]$.
	
	Recall that (by construction) algorithms $\hat\eta$ and $\pi^*$ select the same source $m^*(t)$ at every slot $t$ and (due to the coupling argument in Remark~\ref{r1}) when a transmission by $\hat\eta$ is successful, the transmission by $\pi^*$ is also successful. Hence, it follows that algorithm $\pi^*$ also delivered the latest packet generated from source $m$ during (or before) slot $t'_m$,
	%
	which implies that $\tau_m^{\hat\eta}(s_p) = \tau_m^{\pi^*}(s_p)$ 
	and, as a result, we have
		$h^{\hat\eta}_m(s_p) = s_p - \tau_m^{\hat\eta}(s_p) = s_p - \tau_m^{\pi^*}(s_p) = h^{*}_m(s_p)$.
	
	Since $h^{\hat\eta}_m(s_p) = h^{*}_m(s_p)$ for every source $m$ and for the first slot of every period $p$, it follows that if period $p$ is an optimal period (in which $\hat\eta$ and $\pi^*$ select the same source and the same channel in every slot $t$) then 
	$h^{\hat\eta}_m(t)=h^{*}_m(t)$ in every slot $t$ within period $p$ and for every source $m$.
\end{proof}

From Lemma \ref{newl1}, we have that if period $p$ is an optimal period, then
\begin{equation}
	\sum_{m=1}^{M} \sum_{t=s_p}^{f_p} h^{\hat\eta}_m(t) = \sum_{m=1}^{M} \sum_{t=s_p}^{f_p} h^{*}_m(t) \; .
\end{equation}
Using this result in the expression of the auxiliary AoI regret, we obtain
\begin{align}
	\mathrm{E}&\left[ \sum_{m=1}^{M} \sum_{t=1}^{T} \left[ h^{\hat\eta}_m(t) - h^{*}_m(t)\right] \right] = \nonumber\\
	&= \mathrm{E}\left[ \sum_{m=1}^{M} \sum_{t=1}^{T} \underbrace{ \left[ h^{\hat\eta}_m(t) - h^{*}_m(t)\right] }_{\le h^{\hat\eta}_m(t)} \mathbbm{1}\left\{n(t)\neq n^* \cap E(t)=0 \right\}\right.+\nonumber\\
	&\qquad + \left. \underbrace{ \left[ h^{\hat\eta}_m(t) - h^{*}_m(t)\right] \mathbbm{1}\left\{n(t)=n^*  \cup E(t)=1 \right\} }_{=0} \right]\nonumber
\end{align}	
\begin{align}
	&\overset{(a)}{\le} \mathrm{E}\left[ \sum_{m=1}^{M} \sum_{p=1}^{T} y^{\hat\eta}_m(p)\ \mathbbm{1}\left\{\bar n(p)\neq n^*\right\} \right]\nonumber\\
	&= \sum_{m=1}^{M} \sum_{p=1}^{T} \mathrm{E}\left[ y^{\hat\eta}_m(p)\ \mathbbm{1}\left\{\bar n(p)\neq n^*\right\} \right]\nonumber\\
	&\overset{(b)}{=} \sum_{m=1}^{M} \sum_{p=1}^{T} \mathrm{E}\left[ y^{\hat\eta}_m(p) \;\middle\vert\; \bar n(p)\neq n^* \right] \mathrm{P}\left(\bar n(p)\neq n^*\right)
\end{align}
where (a) follows from the fact that each period $p$ has duration of at least $1$ time slot, and (b) follows from the law of total expectation. 

\section{Proof of Lemma \ref{step3}}\label{newC}

To establish Lemma~\ref{step3}, we first use Hoeffding's inequality to upper bound $\mathrm{P}\left(\bar n(p) = n\right)$ by an exponential function of $-p$, for every suboptimal channel $n$. The result then follows directly from this upper bound.

\begin{lemma}\label{hoeffding}
	For every suboptimal channel index $n\neq n^*$, the probability of channel policy $\bar\eta_b$ selecting channel $n$ in the nonempty phase of period $p$ is bounded by
	\begin{equation}
		\mathrm{P}\left( \bar n(p)=n \right) \le 2\exp\left(-\dfrac{1}{2N^2}p\right) + 2\exp\left(-\dfrac{\Delta_n^2}{4N}p\right) \; ,
	\end{equation}
	where $\Delta_n = \mu^* - \mu_n$.
\end{lemma}
\begin{proof}
	The channel policy $\bar\eta_b$ uses the empty phases of each period $p$ to explore the channels and update its estimate of the channel reliabilities. We know that at the beginning of each period $p$ there is an empty phase with at least one exploration slot. Let $N_n(p)$ be the number of exploration slots in which channel $n$ is selected within the first $p$ periods and let $p'$ be the total number of exploration slots within the first $p$ periods. It follows that
	\begin{equation}
		\sum_{n=1}^{N} N_n(p) = p' \ge p \; \mbox{ and } \; \mathrm{E}\left[ N_n(p) \right] = \dfrac{1}{N}p' \; .
	\end{equation}
	
	
	
	Let $\hat\mu_n(p')$ be the estimate of the reliability of channel $n$ after a total of $p'$ exploration slots. Then, for any suboptimal channel $n\neq n^*$, we have
	\begin{align}
		&\mathrm{P}\left( \bar{n}(p)=n \right) \le \mathrm{P}\left( \hat\mu_n(p') \ge \hat\mu^*(p') \right)\nonumber\\
		&\quad \overset{(a)}{\le} \mathrm{P}\left( \hat\mu_n(p')\ge\dfrac{\mu_n+\mu^*}{2} \right) + \mathrm{P}\left( \dfrac{\mu_n+\mu^*}{2}\ge\hat\mu^*(p') \right)\nonumber\\
		&\quad = \mathrm{P}\left( \hat\mu_n(p')-\mu_n\ge\dfrac{\Delta_n}{2} \right) + \mathrm{P}\left( \hat\mu^*(p')-\mu^* \le -\dfrac{\Delta_n}{2} \right) \; .\label{newprobbound1}
	\end{align}
	where (a) follows from the union bound. Denote $\hat\mu_n(p')-\mu_n=I_n$. Then,
	\begin{align}
		&\mathrm{P} \left( \hat\mu_n(p')-\mu_n\ge\dfrac{\Delta_n}{2} \right)\nonumber\\
		& = \mathrm{P}\left( I_n\ge\dfrac{\Delta_n}{2} ,\, N_n(p)\le\dfrac{p'}{2N} \right) + \mathrm{P}\left( I_n>\dfrac{\Delta_n}{2} ,\, N_n(p)>\dfrac{p'}{2N} \right)\nonumber\\
		& \le \mathrm{P}\left( N_n(p)\le\dfrac{p'}{2N} \right) + \mathrm{P}\left( I_n>\dfrac{\Delta_n}{2} \;\middle\vert\; N_n(p)>\dfrac{p'}{2N} \right) \; .\label{newhoeffding1}
	\end{align}
	Next, we upper bound each of the last two terms by Hoeffding's inequality. 
	
	Notice that:
	\begin{itemize}
		\item $N_n(p)$ is the sum of $p'$ i.i.d.\ Bernoulli random variables with mean $\dfrac{1}{N}$; and
		\item $\hat\mu_n(p')$ is the average of $N_n(p)$ i.i.d.\ Bernoulli random variables with mean $\mu_n$.
	\end{itemize}
	
	Thus, by Hoeffding's inequality
	\begin{align}
		\mathrm{P}\left( N_n(p)\le\dfrac{1}{2N}p' \right) &= \mathrm{P}\left( \dfrac{N_n(p)}{p'}-\dfrac{1}{N}\le-\dfrac{1}{2N} \right)\nonumber\\
		&\le \exp\left(-\dfrac{1}{2N^2}p'\right)\label{newhoeffding2}
	\end{align}
	and
	\begin{align}
		\mathrm{P}\left( I_n>\dfrac{\Delta_n}{2} \;\middle\vert\; N_n(p)>\dfrac{1}{2N}p' \right)
		&\le \exp\left(-\dfrac{\Delta_n^2}{4N}p'\right) \; .\label{newhoeffding3}
	\end{align}
	Inequalities \eqref{newhoeffding1}, \eqref{newhoeffding2} and \eqref{newhoeffding3} imply
	\begin{align}
		&\mathrm{P} \left( \hat\mu_n(p')-\mu_n\ge\dfrac{\Delta_n}{2} \right) \le \exp\left(-\dfrac{1}{2N^2}p'\right) + \exp\left(-\dfrac{\Delta_n^2}{4N}p'\right).\label{newprobbound2}
	\end{align}
	Analogously, we have
	\begin{align}
		&\mathrm{P} \left( \hat\mu^*(p')-\mu^*\le\dfrac{\Delta_n}{2} \right) \le \exp\left(-\dfrac{1}{2N^2}p'\right) + \exp\left(-\dfrac{\Delta_n^2}{4N}p'\right) \; .\label{newprobbound3}
	\end{align}
	Now, inequalities \eqref{newprobbound1}, \eqref{newprobbound2} and \eqref{newprobbound3} imply that
	\begin{align}
		\mathrm{P}\left( \bar{n}(p')=n \right) &\le 2\exp\left(-\dfrac{1}{2N^2}p'\right) + 2\exp\left(-\dfrac{\Delta_n^2}{4N}p'\right)\nonumber\\
		&\le 2\exp\left(-\dfrac{1}{2N^2}p\right) + 2\exp\left(-\dfrac{\Delta_n^2}{4N}p\right)
	\end{align}
\end{proof}
Using Lemma~\ref{hoeffding}, we obtain
\begin{align}
	&\sum_{p=1}^{T} \mathrm{P}\left(\bar n(p)\neq n^*\right) = \sum_{n\neq n^*} \sum_{p=1}^{T} \mathrm{P}\left(\bar n(p)= n\right)\nonumber\\
	&\quad \le \sum_{n\neq n^*} \sum_{p=1}^{T} \left[ 2\exp\left(-\dfrac{1}{2N^2}p\right) + 2\exp\left(-\dfrac{\Delta_n^2}{4N}p\right) \right]\nonumber\\
	&\quad \le \sum_{n\neq n^*} \sum_{p=1}^{\infty} \left[ 2\exp\left(-\dfrac{1}{2N^2}p\right) + 2\exp\left(-\dfrac{\Delta_n^2}{4N}p\right) \right]\nonumber\\
	&\quad = \sum_{n\neq n^*} \left[ \dfrac{2}{\exp\left(\dfrac{1}{2N^2}\right)-1} + \dfrac{2}{\exp\left(\dfrac{\Delta_n^2}{4N}\right)-1} \right]\nonumber\\
	&\quad \le \dfrac{2(N-1)}{\exp\left(\dfrac{1}{2N^2}\right)-1} + \dfrac{2(N-1)}{\exp\left(\dfrac{\Delta_{min}^2}{4N}\right)-1}
\end{align}
which is a constant, where $\Delta_{min} = \mu^* - \max\limits_{n\neq n^*}\mu_{n}$. 

\end{document}